\documentclass[12pt]{amsart}

\usepackage{amssymb}
\usepackage{amscd}
\usepackage{enumerate}
\usepackage{graphicx}
\usepackage{color}

\numberwithin{equation}{section}

\theoremstyle{plain}

\newtheorem{theorem}{Theorem}[section]
\newtheorem{proposition}[theorem]{Proposition}
\newtheorem{lemma}[theorem]{Lemma}

\newtheorem{conjecture}[theorem]{Conjecture}

\newtheorem{question}[theorem]{Question}

\newcommand\R{{I\!\!R}}

\DeclareFontFamily{OT1}{rsfs}{}
\DeclareFontShape{OT1}{rsfs}{n}{it}{<-> rsfs10}{}
\DeclareMathAlphabet{\mathscr}{OT1}{rsfs}{n}{it}

\begin{document}

\title{Some open questions in hydrodynamics}

\author{Mateusz Dyndal}
\address{AGH Univ. of Science and Technology, Cracow, Poland \\
         CEA Saclay, Irfu/SPP, Gif-sur-Yvette, France}
\email{mateusz.dyndal@cern.ch}

\author{Laurent SCHOEFFEL}
\address{CEA Saclay, Irfu/SPP, Gif-sur-Yvette, France}
\email{laurent.olivier.schoeffel@cern.ch}

\begin{abstract}  
When speaking of  unsolved problems in physics, this is surprising at first glance to discuss the case of 
fluid mechanics. However, there are many deep open questions that come with the theory of fluid mechanics.
In this paper,
we discuss some of them that we classify in two categories, the long term behavior of solutions of equations of hydrodynamics and
the definition of initial (boundary) conditions. The first set of questions come with the non-relativistic theory
based on the Navier-Stokes equations. Starting from smooth initial conditions, the purpose is to understand if solutions of 
Navier-Stokes equations remain smooth with the time evolution.
 Existence for
just a finite time would imply the evolution of finite time singularities, which would have
a major influence on the development of turbulent phenomena. 
The second set of questions come with the relativistic theory of hydrodynamics.
There is an accumulating evidence that this theory may be relevant for the description of the medium 
created in high energy heavy-ion collisions. However, this is not clear that the fundamental hypotheses of
hydrodynamics are valid in this context. Also, the determination of initial conditions
remains questionable. 
The purpose of this paper is to explore some ideas related to these questions, both in the
non-relativistic and relativistic limits of fluid mechanics.
We believe that these ideas do not concern only the theory side but can also be useful for interpreting results from experimental measurements.

\end{abstract}

\maketitle

\section{Introduction}

Equations of fluid mechanics are widely used by engineers and physicists with a great deal of success to describe 
  many physical phenomena in nature. For example,
they can be used to model large scale systems like the coupled atmospheric and ocean flows used by
the meteorological office for weather prediction down to small scale systems in chemical
engineering. 
In the non-relativistic limit, 
the derivation of these equations for ideal fluids dates back to  Euler (1755), and
 for a viscous fluid to Navier (1822) and Stokes (1845).
Similarly, for relativistic motions, equations of fluid mechanics have 
become of increasing importance in the context of high-energy
physics. This dates back to the seminal works of Fermi and Landau (1950) \cite{landau}, dedicated to high-energy collisions  
producing  many hadrons of different
sorts going into all directions. 
 With the advent of heavy-ion collisions,  the interest in relativistic fluid mechanics has been recently
revived after the observation of several phenomena 
which suggested that the matter produced in these collisions could behave
collectively, like a fluid \cite{Heinz:2001xi,Gyulassy:2004vg}. 

Despite the great success of hydrodynamics over history to represent and simulate fluid behaviors, 
in the non-relativistic limit, we do not
know yet if this description meets fundamental mathematical criteria. 
In the following we define a smooth function as  infinitely differentiable on $\R^3$, square integrable
and with strong decay properties at infinity for the function itself and its derivatives.
In particular, the question is still open on
whether solutions to equations of fluid mechanics in $\R^3$ could develop singularities or not, under smooth 
divergence free initial conditions.
The answer to this important question is
recognized as one of the Millennium prize problems
\cite{fefferman:2006,ladyzhenskaya:2003}.
If they do, our description of small scale flows must be missing essential
physics.
The purpose of this paper is to introduce the underlying mathematical problem in simple and
 precise terms. 
We begin by introducing the incompressible Navier-Stokes  equations and then 
define a few open  questions concerning the global regularity of solutions.
When discussing these questions,  our purpose to to give pedagogical arguments that are readable even by non experts 
but that nevertheless contain all the ideas.
In the same way, we present some selected key open questions in relativistic hydrodynamics.
For this part, we focus the discussion on what could happen during high energetic collisions of heavy ions, as this
has been experienced at RHIC and currently at the LHC.
A summary of the experimental situation can be found in \cite{Huovinen:2013wma}.
Our purpose here is not to cover all aspects of this physics but to focus on the applicability of
hydrodynamics in this context. Fluid mechanics is rooted on fundamental hypotheses. If their validity is not debatable for
standard non-relativistic flows, this could be the case for relativistic flows possibly generated in
heavy-ion collisions. In particular, when a very large number of elementary constituents are produced locally at
high energy, this is unclear whether pure statistical and geometrical descriptions could be sufficient
or not in describing the possible outcomes of experiments.

Obviously, there are some links between the two parts of this paper, related to non-relativistic and relativistic limits of
hydrodynamics.  This is also our purpose to describe these links by a global presentation of hydrodynamics.
For example, the role of viscosity is essential, maybe for damping singularities that could arise from 
ideal fluid mechanics equations. Viscosity is  a measure of the ability of a fluid to return to its local thermodynamic
equilibrium. Then, when strong interactions among 
elementary constituents need to be considered, we can expect small relaxation times and thus small viscosity effects.
This can drive some parallel thinking between the two cases, relativistic or not.

\section{First hypotheses}
\label{hyp}

The formalism of fluid mechanics requires thermodynamics because a large part of the information on the flow is encoded
in its thermodynamic properties, trhough its equation of state. This relies on  two
assumptions: the continuum hypothesis and the  local thermodynamic
equilibrium. No other assumption is made concerning the nature of the
particles and fields, their interactions, the classical or quantum nature
of the phenomena involved. 

Indeed, hydrodynamics, relativistic or not, is supposed to describe motions of fluids and related phenomena at macroscopic scales, which
assumes that a fluid can be regarded as a continuous medium. 
This is the statement of the continuum hypothesis.
This means that any small volume element in
the fluid is always supposed so large that it still contains a very great number of elementary constituents. Accordingly, when
we consider infinitely small elements of volume, we mean very small compared with the volume of the system
under consideration, but large compared with the distances between the elementary constituents. The expressions fluid
particle and point in a fluid are to be understood in this sense. 

Following the continuous assumption, the mathematical description of the state of a moving fluid element can be
characterized by functions of the coordinates $x, y, z$ and of the time $t$. These functions of $(x, y, z, t)$ are related
to the quantities defined for the fluid at a given point $(x, y, z)$ in space and at a given time $t$, which refers to
fixed points in space and not to fixed particles of the fluid. This is called the Eulerian description of the fluid flow, which we will commonly use in writing equations hereafter. In addition, we always assume that the fluid is close to 
local thermodynamic equilibrium, which
 means that its thermodynamic properties are varying slowly from point to point.
 This requires  that local relaxation times towards thermal equilibrium are much shorter than any
macroscopic dynamical time scale. In particular, microscopic collision time scale (between
elementary constituents of the fluid) needs to be much shorter than any macroscopic evolution time
scales. 
This hypothesis is almost a tautology for non-relativistic fluids, 
but becomes not trivial in the relativistic case.

We recall the expression of the differential of the 
internal energy $U$ for the system under consideration, of mass $m$ and volume $V$ 
\begin{equation}
dU=-PdV+TdS+\mu dN,
\end{equation}
where $P$ is the pressure, 
$S$ is the entropy, $T$ the temperature and $\mu$ the chemical potential. 
In non-relativistic systems, $N$ is generally the number of  molecules.
This term vanishes as the chemical reactions in the
fluid element are frozen with respect to the dynamical time scale of the flow. There
is a local thermodynamic equilibrium for all species involved in the chemical reactions. 
  In relativistic 
systems, the number of elementary particles is not conserved as it is always possible
to create a particle-antiparticle pair, provided energy is available. Therefore, in this case, $N$ does not
represent a number of particles, but a conserved quantity, such as the baryon number. 

In the following sections, it will be useful to introduce 
densities per unit volume or per unit mass. 
We note energy density per unit volume $\epsilon= U/V$, per unit mass $\epsilon_m= U/m$ and the entropy 
density per unit volume $s= S/V$, per unit mass $s_m= S/m$. The mass density can then be written as $\rho=m/V$.

\section{Equations of Non-relativistic  fluid mechanics}
\label{nonrel}

The mathematical description of the state of a moving fluid is effected by its velocity field $\vec{v}(x,y,z,t)$
and any two thermodynamic quantities pertaining to the fluid, for example the pressure 
$P(x,y,z,t)$and the mass density $\rho(x,y,z,t)$. 
The velocity field $\vec{v}$ is a vector with three components $(v_x, v_y, v_z)$ each of which may be functions of $(x,y,z,t)$.
Note that thermodynamic quantities are determined by the value of two of them with the equation of state. Hence, with five quantities, the mathematical state of the moving fluid is completely determined.

We first recall the equations relating these variables in the non-relativistic approximation. These equations are related to  three conservation principles: (1) the conservation of mass, (2) the balance of momentum or the Newton's second law of mechanics and (3) the conservation of energy.
For ideal fluids, these conservation laws read in the absence of external force
\begin{equation}
\frac{\partial\rho }{\partial t}+\vec\nabla(\rho \vec v)=0 ;
\label{mass}
\end{equation}
\begin{equation}
	\frac{\partial}{\partial t} \vec{v} + \left( \vec{v}\cdot\nabla \right) \vec{v}  = -\frac{1}{\rho} \nabla P ;
\label{euler}
\end{equation}
\begin{equation}
\frac{\partial}{\partial t} S + \left(\vec{v}\cdot\nabla \right) S  =0.
\label{isos}
\end{equation}
For simplicity of the notations, we have not written explicitly the dependencies in
$(x,y,z,t)$ for all functions in the above relations.
Equation (\ref{mass}) is called the continuity equation.
Equations (\ref{euler}) correspond to three partial differential equations for the three components of the velocity field. They are called the Euler's equations (for ideal fluids).
The hypothesis behind ideal fluids is that each particle pushes its neighbors equally in every
direction. This is why a single scalar quantity, the pressure, is sufficient to describe the force per unit area
that a particle exerts on all its neighbors at a given time. Then, the acceleration of the fluid particle results
from the pressure differences. This is what Euler's equations (\ref{euler}) describes mathematically.
Equation (\ref{isos}) represents the conservation of entropy of the fluid throughout its evolution in time.
Indeed, for ideal fluids, we neglect all processes related to energy dissipation, which may occur in a moving fluid
as a consequence of internal friction (viscosity)  as well as heat exchange between different
parts of the fluid.
Therefore, the motion of an ideal fluid is by definition considered as adiabatic. Consequently, the entropy of any fluid particle remains constant as that particle moves in space inside the fluid.
This condition usually takes a
much simpler form. As it usually happens, if the entropy is constant throughout some volume element of
the fluid at some initial time, then it retains the same constant value everywhere in the fluid volume, at all
times for any subsequent motion of the fluid. In this case, the conservation of entropy (\ref{isos}) reads
$$
S=constant.
$$
This is the condition for an isentropic motion.
Note that the variation enthalpy per unit mass ($h_m=H/m$) for isentropic motion reads
$$
dh_m=Tds_m+\frac{1}{\rho}dP=\frac{1}{\rho}dP.
$$
Then, equations (\ref{euler}) become
\begin{equation}
	\frac{\partial}{\partial t} \vec{v} + \left( \vec{v}\cdot\nabla \right) \vec{v}  = - \nabla h_m.
\label{euler2}
\end{equation}
The continuity equation and the Euler's equations,
together with the equation of state  $P=P(\rho,s_m)$ form a system of five equations  for the set of five variables: the mathematical system is closed.

For non-relativistic motion, a reasonable approximation is to consider a constant mass density
$\rho = \rho_0$.
It is called the incompressibility condition, well verified for subsonic flows, for which 
the  velocity field is  much smaller in magnitude than the sound speed in the fluid.
In the following, we keep this approximation in the non-relativistic limit. Then,
the continuity equation (\ref{mass}) reads
\begin{equation}
\nabla\cdot \vec{v} = 0.
\label{incomp}
\end{equation}

For real fluid, for which energy dissipation in the fluid is not neglected any longer, the viscosity or internal
friction must be taken into account. A modification of the right hand side of Euler's equations (\ref{euler}) is thus needed and these equations need to be upgraded to the 
Navier-Stokes partial differential equations
\begin{equation}
	\frac{\partial}{\partial t} \vec{v} + \left(\vec{v}\cdot\nabla\right)\vec{v}  = -\frac{1}{\rho_0} \nabla P +\frac{\eta}{\rho_0} \Delta \vec{v}.
\label{ns}
\end{equation}
Here again, we have assumed the absence of external force.
The positive coefficient $\eta$ is called the dynamic viscosity, 
describing the quality of the fluid, while $\nu=\frac{\eta}{\rho_0}$ is called the kinematic viscosity.
In general the viscosity coefficient is a function of pressure and temperature of the fluid. As pressure
and temperature may not be constant throughout the fluid, the viscosity coefficient also may not be
constant throughout the fluid. 
Equations (\ref{ns}) complement Euler's equations (\ref{euler}) with the presence of a diffusive term in
${\nu} \Delta \vec{v}$. The convective term in $\left(\vec{v}\cdot\nabla\right)\vec{v}$ is not affected by
the internal friction phenomenon.
The idea behind the diffusive term is that when energy dissipation in the fluid is not neglected, the viscosity or internal
friction is due to an irreversible transfer of momentum from points where the velocity is large to those
where it is small. This means that if a fluid particle moves faster than the average of its neighbors, then friction
slows it down. 
Knowing 
the
mathematical fact that the deviation of a function $f$ at a point from its average value on small
surrounding spheres is measured by $-\Delta f$, we can conclude immediately
that frictional (vector) forces must be proportional to $\Delta \vec{v}$ at first approximation.
This gives a simple justification of the functional form of Navier-Stokes equations (\ref{ns}).

Interestingly, incompressible Navier-Stokes equations  can be rewritten using dimensionless variables, 
for a flow with 
characteristic scales for the velocity,  $U$, and for length, $L$. Then, after a proper rescaling of 
variables in 
 equation (\ref{ns}), it becomes
\begin{equation}
	\frac{\partial}{\partial t} \vec{v} + \left(\vec{v}\cdot\nabla\right)\vec{v}  = - \nabla P +\frac{1}{Re} \Delta \vec{v},
\label{nsrey}
\end{equation}
where all quantities are dimensionless and $Re=UL/\nu$ is called the Reynolds number. It defines a macroscopic (dimensionless) number corresponding to the ratio of the strength of the non-linear
effects to the strength of the linear viscous effects. 
This expression leads immediately to the law of similarity:
flows which can be obtained from one another by simply changing the unit of measurement of coordinates
and velocities are said similar. Therefore flows of the same type and same Reynolds numbers are similar.
 The most important comment to be stated here concerns the large $Re$ limit. As it can be seen from equation (\ref{nsrey}), when 
$Re \rightarrow \infty$, Navier-Stokes equations (\ref{ns}) tend to Euler's equations (\ref{euler}).
Then, we may think that we can neglect the viscous term the Navier-Stokes equations in
comparison with the convective term when the Reynolds number is sufficiently large (at small
viscosity). However, the Navier-Stokes equations correspond to a singular perturbation of the Euler
equations, since the viscosity multiplies the term that contains the highest-order spatial derivatives.
As a result, this is not always possible to operate this simplification. The high Reynolds number limit
of the Navier-Stokes equations is a very difficult problem, where turbulent effects may dominate the
dynamic of the flow.

Navier-Stokes equations need boundary conditions. Indeed, there are
always forces of molecular attraction between viscous fluid and the surface of a solid body, and these forces
have the result that the layer of fluid immediately adjacent to the surface is brought completely to rest, and
adheres to the surface. Then, if we consider the fluid inside some  domain $\Omega \subset \R^3$, then the fluid particles stick to the walls
of the domain $\partial \Omega$ and $\vec{v}(\vec{x},t)=\vec{0}$ for $\vec{x} \in \partial \Omega$. 

Incompressible Navier-Stokes equations are non-linear and non-local. The non-linearity is due to the convective term. Non-locality refers to the relationship between  velocity  and pressure. From equations (\ref{incomp})  and (\ref{ns}) it can easily be shown that
\begin{equation}
\Delta P = -\sum_{i,j=1}^{3} \frac{\partial v_i}{\partial x_j} \frac{\partial v_j}{\partial x_i}.
\label{deltap}
\end{equation}
Then, any change in the velocity field at a position x affects the pressure field
immediately everywhere else. This implies that sound
waves can carry any perturbations of the pressure field  instantaneously across the entire volume of the
fluid.
It is clear that  we need to understand the interplay between
the non-linear and diffusive terms in the
Navier-Stokes equations (\ref{ns}). Due to non-linearities in the equations, a singularity of the velocity field could arise 
as the result of the action of a flow onto
itself, even when the initial conditions are smooth and divergence free. However, 
friction, if large enough, could prevent the velocity from
becoming singular, if this happens to be the case.
Then, regarding  Navier-Stokes equations, we would need to ensure
(i) the existence of solutions,
the physical system must have a way to evolve into the future, (ii) uniqueness,
there must not be arbitrary choices for the evolution, and (iii) continuous
dependence on the initial state, any future state of the flow is determined, to arbitrary
finite precision, by the initial conditions.
For the incompressible Navier-Stokes equations, a complete
answer to these questions is open. This is also true for Euler's equations.
What is known is that both are locally
well-posed : solutions starting out from smooth initial
data are unique, depend continuously on the initial data, and remain smooth
for at least a finite, possibly short, interval of time.

During its evolution in time, the kinetic energy $K$  of the fluid 
needs also to remain finite. It can be expressed a function of time $t$ as
\begin{equation}
K(t) = \frac{1}{2} \int_{\Omega} \| \vec{v}(x,y,z,t) \|^2 dV,
\end{equation}
where $\| \vec{v} \|^2=\sum_{i=1}^{3} |v_i|^2$ represents the norm (squared) of 
the velocity vector field  and $dV$ is the volume element, $dV=dx dy dz$.
Also, we take the constant mass density $\rho_0$ to be unity and keep this convention for
the discussion that follows for non-relativistic fluid dynamics.
When there is no external source of energy, as assumed here in the absence of external force, the kinetic energy 
dissipates. Then, a fundamental relation holds
\begin{equation}
 \frac{1}{2} \int_{\Omega} \| \vec{v}(x,y,z,t) \|^2 dV + \nu \int_{0}^{t}\int_{\Omega} \| \nabla \vec{v}(x,y,z,s) \|^2 dV ds
\leq \frac{1}{2} \int_{\Omega} \| \vec{v}(x,y,z,0) \|^2 dV,
\label{v2}
\end{equation} 
with the notation 
\begin{equation}
\| \nabla \vec{v} \|^2=
\sum_{i,j=1}^{3} |\frac{\partial v_j}{\partial x_i} |^2.
\label{dv2}
\end{equation} 
The quantity $\nu \int_{0}^{t}\int_{\Omega} \| \nabla \vec{v}(x,y,z,s) \|^2 dV ds$ represents the
cumulative energy dissipation, which measures how much kinetic energy has been lost up to time $t$.
Solutions with finite kinetic energy and with finite average rate of dissipation of kinetic
energy should, in principle, exist forever and decay to zero. Unfortunately, the dissipation of
kinetic energy is the strongest quantitative information about incompressible Navier-Stokes equations that is presently known for general solutions. Then, the finiteness of $K(t)$ once we ensure the finiteness of $K(0)$ is an interesting property of 
the flow but can not prevent the divergence of the velocity field or its derivatives at a certain time $T^*$. In such case,
we call of blow-up at $T^*$. Following this discussion, some open questions come immediately.
\begin{question}
What are the most general conditions for smooth initial incompressible velocity field
$\vec{v}(\vec{x},0)$ that ensures that  smooth solutions of the  Navier-Stokes equations 
exist for $t \geq 0$? The same question holds for Euler's equations.
\label{nonrel1}
\end{question}
\begin{question}
Can we understand physically the mechanism that could generate finite time blow-up of solutions
of the Navier-Stokes equations starting from smooth initial incompressible velocity field?
\label{nonrel2}
\end{question}
\begin{question}
Following the previous questions, which properties of the velocity field 
would be needed to guarantee the existence of a (unique) global smooth solution
at all times, once the initial conditions are smooth?
\label{nonrel3}
\end{question}
Of course, all these questions are not new. They are heavily discussed in the mathematical literature.
Our purpose is to give pedagogical and precise arguments that can be read by anyone in order to
 get a better understanding of non-relativistic fluid mechanics.
This will be done in section \ref{regular}.
We will use extensively the short-hand notations, written below for $\Omega=\R^3$
\begin{eqnarray}
\| \vec{v} \|^2_{L^2(\R^3)}(t) &=& \int_{\R^3} \|  \vec{v}(x,y,z,t) \|^2 dV \\
\| \nabla \vec{v} \|^2_{L^2(\R^3)}(t) &=& \int_{\R^3} \| \nabla  \vec{v}(x,y,z,t) \|^2 dV,
\label{notationl2}
\end{eqnarray}
where $\| \vec{F} \|^2_{L^2(\R^3)}(t)$
represents the $L^2(\R^3)$ norm of the vector field $\vec{F}(x,y,z,t)$.

\section{Equations of relativistic  fluid mechanics}
\label{rel}
In the relativistic limit, it is useful to define the 4-velocity $u^\mu$,
which transforms as a  4-vector under Lorentz transformations. It reads
\begin{equation}
\begin{bmatrix}{u^0}\\ \\ {\vec u}\end{bmatrix}=
\begin{bmatrix}{{1}/ {\sqrt{1-\vec v^2}}}\\ \\{{\vec v}/ {\sqrt{1-\vec v^2}}}\end{bmatrix},
\label{umu}
\end{equation}
where we have chosen a  system of units where $c=1$.
Then, $u^\mu u_\mu=(u^0)^2-\vec u^2=1$. In the following, we also pose $k_B=1$ and $\hbar=1$.

In relativistic fluid mechanics, the non-relativistic conservation of mass (\ref{mass}) is transformed into
\begin{equation}
\partial_\mu (n u^\mu)=0,
\label{nconservation}
\end{equation}
where we have used the standard notation $\partial_\mu =\partial/\partial x^\mu$. 
Equation (\ref{nconservation}) is a conservation equation for
the 4-vector $n u^\mu$,
where $n$ represents the baryon density defined in the fluid rest frame. The baryon density is then
$n u^0$ and its flux $n \vec u$. 

Following the same program as in the non-relativistic case,
the conservation of total energy and momentum gives 4  conservation equations.
They read
\begin{equation}
\partial_\mu T^{\mu\nu}=0,
\label{dmutmunu}
\end{equation}
where $T^{\mu\nu}$ is the energy-momentum tensor. For ideal fluids, it is given by
\begin{equation}
T^{\mu\nu}=(\epsilon +P)u^\mu u^\nu-P g^{\mu\nu},
\label{tmunu}
\end{equation}
where $g^{\mu\nu}\equiv {\rm diag}(1,-1,-1,-1)$ is the Minkowski metric tensor.
We recall that $\epsilon$ is the internal energy density per unit volume, which includes the mass
energy density in the relativistic formulation.

The different components of the energy-momentum tensor can be understood as follows
\begin{itemize}
\item $T^{00}$ is the energy density;
\item $T^{0j}$ is the density of the $j^{th}$ component of momentum, with $j=1,2,3$;
\item $T^{i0}$ is the energy flux along axis $i$;
\item $T^{ij}$ is the flux along axis $i$ of the $j^{th}$ component of momentum.
\end{itemize}

Also,
we can easily check that $T^{\mu\nu}$ reduces to a simple matrix $T_{(0)}$ in the rest frame of the 
fluid, where  $u^\mu=(1,0,0,0)$. We obtain
\begin{equation}
T_{(0)}=\left(\begin{array}{cccc}\epsilon &0&0&0\\ 0&P&0&0\\ 0&0&P&0\\ 0&0&0&P\end{array}\right)
\label{tmunu0}
\end{equation}

Indeed, in the fluid rest frame, the assumption of local thermodynamic equilibrium implies isotropy. Hence,
 the energy flux $T^{i0}$ 
and the momentum density $T^{0j}$ vanish. In addition, this implies that the pressure 
tensor is proportional to the identity matrix, $T^{ij}=P\delta_{ij}$, which leads trivially to equation (\ref{tmunu0}).

Equations (\ref{nconservation}),  (\ref{dmutmunu}) and (\ref{tmunu})  represent the equations for 
ideal  relativistic fluid mechanics, in the absence of viscosity effects. Together with the equation of state of the fluid,
they form 
a closed system of equations.

Interestingly, we can rewrite the conservation equations (\ref{dmutmunu}) in terms of 
the velocity field $\vec{v}$ and the pressure. This will be useful for a direct comparison with
non-relativistic ideal fluid mechanics equations of section \ref{nonrel}. We get
\begin{equation}
\frac{\partial P}{\partial t} - \frac{\partial }{\partial t} 
\left( \frac{\epsilon + P}{1-\vec{v}^2} \right) -\nabla \cdot \left( \frac{(\epsilon + P)\vec{v}}{1-\vec{v}^2} \right)=0,
\label{enrel}
\end{equation}
for  energy conservation and
\begin{equation}
\frac{\partial}{\partial t} \vec{v} + \left( \vec{v}\cdot\nabla \right) \vec{v}  
= -\frac{1-\vec{v}^2}{\epsilon + P} \left[ \nabla P 
+ \vec{v} \frac{\partial P}{\partial t} \right],
\label{eulerrel}
\end{equation}
for momentum conservation. In order to get expressions as written here (\ref{eulerrel}) for the space
component of equations (\ref{dmutmunu}), $\partial_\mu T^{\mu i}=0$, we can use equation (\ref{enrel}).
Using the notation $D=u^\mu \partial_\mu$, these equations can be written in a more compact form as 
\begin{equation}
D \epsilon = -(\epsilon + P) \nabla \vec{u},
\label{enrel2}
\end{equation}
\begin{equation}
D u^\nu = -\frac{\nabla^\nu P}{(\epsilon + P)}.
\label{eulerrel2}
\end{equation}

At this stage, the link with section \ref{nonrel} can be done by taking
the non-relativistic limit of these equations, in which $\epsilon$ is then dominated 
by the mass energy density,  $\epsilon \simeq \rho$. Then,
equations (\ref{eulerrel}) recover their trivially non-relativistic form (\ref{euler}). In order to take the limit of 
equation (\ref{enrel}), we can write the internal (relativistic) energy density $\epsilon$
as its sum in terms of the mass energy density and a term not including the mass that we label as $\epsilon_{int}$. Then,
$$
\epsilon = \rho +\epsilon_{int}.
$$
Equation (\ref{enrel}) becomes
$$
\frac{\partial \rho}{\partial t} +\nabla \cdot \left( \rho \vec{v} \right)+
\frac{\partial \epsilon_{int}}{\partial t} +\nabla \cdot \left( (\epsilon_{int} + P)\vec{v} \right)=0.
$$
The first part of the equation vanish due to mass conservation. We are left with
\begin{equation}
\frac{\partial \epsilon_{int}}{\partial t} +\nabla \cdot \left( (\epsilon_{int} + P)\vec{v} \right)=0,
\label{en}
\end{equation}
where $\epsilon_{int}$ is the energy density as defined in the non-relativistic approximation.
This last relation (\ref{en}) corresponds to the local form of the energy conservation (in the non-relativistic limit), which can be
shown to be equivalent to the conservation of entropy (\ref{isos}).
Interestingly, in the relativistic case, we can derive the conservation of entropy density, based on equations
(\ref{enrel}) and (\ref{eulerrel}), provided that the following thermodynamic relation holds
\begin{equation}
\epsilon+P=\rho \frac{d\epsilon}{d\rho},
\label{cond}
\end{equation}
where $\rho$ represents the particle density.
The proof is simple. If we define the fluid current density as $j^\mu=\rho u^\mu$. Then, equations (\ref{enrel}) and (\ref{eulerrel})
lead to
\begin{equation}
j^\mu \left [ P \partial_\mu (\frac{1}{\rho}) + \partial_\mu (\frac{\epsilon}{\rho}) \right ]=0.
\label{toto1}
\end{equation}
Also, we have the thermodynamic identity
\begin{equation}
T ds= P d (\frac{1}{\rho}) + d (\frac{\epsilon}{\rho}),
\label{toto2}
\end{equation}
where $s$ is the entropy density (per unit volume).
From equations (\ref{toto1}) and (\ref{toto2}), we obtain
$$
u^\mu \partial_\mu s = \frac{\partial s}{\partial t} +\left(\vec{v}\cdot\nabla\right) s =0,
$$
which completes the proof. 

We can now discuss the physics content of equations (\ref{enrel}) and (\ref{eulerrel}), the relativistic counterparts of the
conservation of mass and Euler's equations of section \ref{nonrel}.
As mentioned in the introduction, we restrict our presentation to the physics case of particle production in high energy heavy-ion 
collisions,  with an energy in the center-of-mass frame per
nucleon pair much larger that the mass of the nucleon $\sqrt{s_{NN}} \gg m_N \simeq 1$~GeV.
A pedagogical review of this topic can be found in \cite{Ollitrault:2008zz}.
It is commonly accepted that these experiments can be understood as a three steps process: (i) the collision creates a small and very dense state of matter which later (ii) undergoes an hydrodynamic expansion and finally (iii) decays into the observed hadrons.
Our purpose is thus to see how equations (\ref{enrel}) and (\ref{eulerrel}) can be used in the description of step (ii).
Of course, the problem of the global regularity of solutions of these equations holds as for the non-relativistic case.
However, in the context of the collisions described just above, this is not the most important aspect. Indeed, in such experiments,
the expansion step (ii) will not last long and thus
we are not interested in the long term evolution of the velocity field from smooth  initial conditions.
In such experiments, the difficult issues are to define what are the most plausible initial conditions, whether or not a thermodynamic equilibrium can be reached, with which equation of state, and then and only then how to solve the coupled equations
(\ref{enrel}) and (\ref{eulerrel}) up to step (iii). This is not clear also if there is a well defined frontier between the decoupling stage (iii) and the 
hydrodynamical evolution period (ii).
In addition, if the step (ii) is really driven by fluid mechanics, the effects of viscosity need to be considered in the
relativistic case. 
The following questions come from this discussion.
\begin{question}
We know from section \ref{hyp} that applicability of hydrodynamics relies on the existence of a local thermodynamic
equilibrium and thus on the condition that
the mean free path of elementary constituents $\lambda_{mfp}$ is much smaller than a dimension scale of the system
under consideration $L$. Can step (i) create this condition that could prevail at step (ii)?
\label{rel1}
\end{question}
\begin{question}
How can we define the most plausible initial conditions at the beginning of step (ii)?
\label{rel2}
\end{question}
\begin{question}
Clearly, these initial conditions will contain the physics of the state of matter created during the collision, through 
the equation of state. This is where the physics comes in. Then, the coupling of the equation of state and the equations of fluid mechanics will govern the expansion of matter at stage (ii) up to (iii). However, this program can only be achieved with another essential part of the initial conditions related to the geometry of the collision. How this geometry can be inferred at step (iii) once the 
 observed hadrons are formed. The immediate question that follows is whether the experimental observables does depend mainly 
on the geometry of the collision or not. Then, the sensitivity of experiments to the physics of the state of matter posed at step (ii)
would be limited.
\label{rel3}
\end{question}

\section{On the existence of global smooth solutions of Navier-Stokes equations}
\label{regular}
In this section,
we discuss the long term behavior of the velocity field for non-relativistic fluids.
Unless stated otherwise, we consider that the fluid is contained in $\R^3$.
We also assume smooth initial conditions
of the velocity field
(divergence free, infinitely differentiable, square integrable and with strong decay properties at infinity) .
Then, question \ref{nonrel1} is related to 
the problem of whether the three-dimensional incompressible
Navier-Stokes equations can develop a finite time singularity from these
smooth initial conditions or if a (unique) global smooth 
solution exists, which means smooth for all times. 
In particular, if it exists, such a solution should be square integrable and thus possess a finite kinetic energy at all times.
This problem is still
unresolved 
\cite{temam:1984,temam:1988,constantin-foias:1988,majda-bertozzi:2002}. 
Note also that the answer to this important question is
recognized as one of the Millennium problems
\cite{fefferman:2006,ladyzhenskaya:2003}.
However, even if we can not solve definitely question \ref{nonrel1}, and thus the related
Millennium problem,
some
partial results can be derived, which are helpful to get a better view of this problem \cite{laurent}.

One key item for this is the finiteness of $\| \vec{v} \|^2_{L^2(\R^3)}(t)$ at all times, 
given by equation (\ref{v2}). This corresponds to solutions with bounded kinetic energy (integrated over all space)
at all times, as it must be.
Another key item is to search for a similar bound associated to the first derivative of the
velocity field by evaluating $\| \nabla \vec{v} \|^2_{L^2(\R^3)}(t)$.
We already know that $\| \vec{v} \|^2_{L^2(\R^3)}(t)$ is bounded at all times. If it was possible to prove that
$\| \nabla \vec{v} \|^2_{L^2(\R^3)}(t)$ is also bounded at all times, then we would have the proof of the existence 
of a global smooth  solution of the Navier-Stokes equations, not blowing up at any time $t \ge 0$ \cite{existence}.
Of course, we know that it will not be possible. Nevertheless, we need to find how $\| \nabla \vec{v} \|^2_{L^2(\R^3)}(t)$
can be bounded in order to understand the limitations to the existence of a global smooth solution.

First, we recall briefly how the bound of the kinetic energy (equation (\ref{v2})) can be obtained. 
\begin{lemma}
Solutions of Navier-Stokes equations, with smooth initial conditions  
follow
\begin{equation}
 \frac{1}{2} \frac{d}{dt}  \| \vec{v} \|^2_{L^2(\R^3)}  =- \nu \int_{\R^3} \| \nabla \vec{v} \|^2 dV.
\label{v2b}
\end{equation}
\label{lemma1}
\end{lemma}
\begin{proof} 
This relation is obtained after developing  $ \frac{1}{2} \frac{d}{dt}  \| \vec{v} \|^2_{L^2(\R^3)}$
as
$$
 \frac{1}{2} \frac{d}{dt}  \|  \vec{v} \|^2_{L^2(\R^3)}  = 
\int_{\R^3} { v_j} { \frac{ \partial v_j}{\partial t}} dV,
$$
where we omit the sum operator for repeated indices. At this step, we can replace  $\frac{ \partial v_j}{\partial t} $
by its expression according to the Navier-Stokes equations
$$
\frac{ \partial v_j}{\partial t} + v_k \partial_k v_j +\partial_j P = \nu \partial_m \partial_m v_j,
$$
which gives
$$
 \frac{1}{2} \frac{d}{dt}  \|  \vec{v} \|^2_{L^2(\R^3)}  = 
\int_{\R^3} { v_j} \left[ -v_k \partial_k v_j -\partial_j P +\nu \partial_m \partial_m v_j \right] dV.
$$
After integrations by parts with vanishing velocity at boundaries (infinity) and the incompressibility condition, it is
trivial to show that
$
 \int_{\R^3} { v_j} \left[ -v_k \partial_k v_j -\partial_j P  \right] dV=0,
$
which completes the proof.
\end{proof}

Similarly, we can evaluate a bound for $\frac{d}{dt}  \|  \nabla \vec{v} \|^2_{L^2(\R^3)}$.
\begin{lemma}
For solutions of Navier-Stokes equations, with smooth initial conditions,
there exists a positive constant $C_\nu$, that depends on $\nu$, such that
\begin{equation}
 \frac{1}{2} \frac{d}{dt}  \|  \nabla \vec{v} \|^2_{L^2(\R^3)} 
\le
-\frac{\nu}{2} \|  \Delta \vec{v} \|^{2}_{L^2(\R^3)}
+C_\nu
 \|  \nabla \vec{v} \|^{6}_{L^2(\R^3)}
\label{dv20}
\end{equation}
\label{lemma2}
\end{lemma}
\begin{proof} 
We proceed as for lemma \ref{lemma1}.  We  write
$$
 \frac{1}{2} \frac{d}{dt}  \|  \nabla \vec{v} \|^2_{L^2(\R^3)}  =
 \int_{\R^3} {\partial_i v_j} {\partial_i \frac{ \partial v_j}{\partial t}} dV,
$$
where  $\frac{ \partial v_j}{\partial t}$ can be replaced by its expression according the the Navier-Stokes equations
$$
 \frac{1}{2} \frac{d}{dt}  \|  \nabla \vec{v} \|^2_{L^2(\R^3)}  =
 \int_{\R^3} {\partial_i v_j} {\partial_i \left[ -v_k \partial_k v_j -\partial_j P +\nu \partial_m \partial_m v_j \right]} dV.
$$
After integrations by parts with vanishing velocity at boundaries (infinity) and the incompressibility condition, 
we obtain
$$
 \frac{1}{2} \frac{d}{dt}  \|  \nabla \vec{v} \|^2_{L^2(\R^3)}  =
 -\int_{\R^3} \partial_i v_j \partial_i v_k \partial_k v_j dV
-\nu \int_{\R^3} \partial_i \partial_i v_j \partial_m \partial_m v_j dV.
$$
This can be rewritten as
\begin{equation}
 \frac{1}{2} \frac{d}{dt}  \|  \nabla \vec{v} \|^2_{L^2(\R^3)} +
\nu \| \Delta \vec{v} \|^2_{L^2(\R^3)} =
 -\int_{\R^3} \partial_i v_j \partial_i v_k \partial_k v_j dV,
\label{dv22}
\end{equation}
using the notation
$$
\| \Delta \vec{v} \|^2_{L^2(\R^3)} = 
\int_{\R^3} \| \Delta \vec{v} \|^2 dV =
\int_{\R^3} \partial_i \partial_i v_j \partial_m \partial_m v_j dV.
$$
Equation (\ref{dv22}) is the key element of the lemma. The difficulty is to evaluate 
a bound of the integral of the right hand side of this equation, which 
contains three times the gradient of the velocity field.
In general, this term has no definite sign.
After some algebra, we can show that
there exists positive constants $C$ and $C_\nu$ such that
$$
-\int_{\R^3} \partial_i v_j \partial_i v_k \partial_k v_j dV
\le 
C
\|  \nabla \vec{v} \|^{3/2}_{L^2(\R^3)}
\|  \Delta \vec{v} \|^{3/2}_{L^2(\R^3)}
\le
\frac{\nu}{2} \|  \Delta \vec{v} \|^{2}_{L^2(\R^3)}
+C_\nu
 \|  \nabla \vec{v} \|^{6}_{L^2(\R^3)},
$$
where the last inequality is a scaled Young's inequality, which
implies that $C_\nu$ is proportional to $1/\nu^3$.
This completes the proof.
\end{proof}

With lemmas \ref{lemma1} and \ref{lemma2},
we have derived the key elements to enter into the discussion of question \ref{nonrel1}.

First, if we consider a fluid in a two-dimensional space $\R^2$, then the problem becomes simple.
Indeed, we have for velocity fields defined on $\R^2$
$$\int_{\R^2} \partial_i v_j \partial_i v_k \partial_k v_j dS=0.$$ 
Then, equation (\ref{dv22}) simplifies to
$$
 \frac{1}{2} \frac{d}{dt}  \|  \nabla \vec{v} \|^2_{L^2(\R^2)} +
\nu \| \Delta \vec{v} \|^2_{L^2(\R^2)} =0.
$$
This trivially implies that $\|  \nabla \vec{v} \|^2_{L^2(\R^2)}(t)$ is bounded at all times and thus
there exists a (unique)
global smooth solution of Navier-Stokes equations in $\R^2$, given smooth initial conditions.
We discuss the uniqueness of the solution, when it exists, below.

Coming back to $\R^3$, lemma \ref{lemma2} can be simplified to obtain the following inequality
$$
 \frac{1}{2} \frac{d}{dt}  \|  \nabla \vec{v} \|^2_{L^2(\R^3)} 
\le
C_\nu
 \|  \nabla \vec{v} \|^{6}_{L^2(\R^3)}.
$$
It can be solved using standard techniques (Gronwall's theorem). We get
$$
\|  \nabla \vec{v} \|^2_{L^2(\R^3)}(t)
\le
\frac{\|  \nabla \vec{v} \|^2_{L^2(\R^3)}(t=0)}
{\sqrt{1-2 C_\nu t \|  \nabla \vec{v} \|^4_{L^2(\R^3)}(t=0)}},
$$
which proves the local existence (finite in time) of a solution of the Navier-Stokes equations 
up to a critical time 
$$
T^*=\frac{1}{2 C_\nu \|  \nabla \vec{v} \|^4_{L^2(\R^3)}(t=0)}.
$$

Also, using the inequality $\|  \nabla \vec{v} \|^2_{L^2(\R^3)}
\le \|   \vec{v} \|_{L^2(\R^3)} \|  \Delta \vec{v} \|_{L^2(\R^3)}$,
which can be derived after integrating by parts,
 lemma \ref{lemma2} can be rewritten as
$$
 \frac{1}{2} \frac{d}{dt}  \|  \nabla \vec{v} \|^2_{L^2(\R^3)} 
\le
-\frac{\nu}{2} \frac{\| \nabla \vec{v} \|^4_{L^2(\R^3)}}{\| \vec{v} \|^2_{L^2(\R^3)}}
+C_\nu
 \|  \nabla \vec{v} \|^{6}_{L^2(\R^3)}
$$
All the quantities appearing in this inequality are of course functions of time.
This is because the positive term
$C_\nu  \|  \nabla \vec{v} \|^{6}_{L^2(\R^3)}$
may be too large compared to
$
-\frac{\nu}{2} \frac{\| \nabla \vec{v} \|^4_{L^2(\R^3)}}{\| \vec{v} \|^2_{L^2(\R^3)}}
$
that we may observe a finite time blow-up of 
$\|  \nabla \vec{v} \|^2_{L^2(\R^3)}$.

It is important to remark that if there exists a time
$T$ for which
$$
-\frac{\nu}{2} \frac{\| \nabla \vec{v} \|^4_{L^2(\R^3)}(T)}{\| \vec{v} \|^2_{L^2(\R^3)}(T)}
+C_\nu
 \|  \nabla \vec{v} \|^{6}_{L^2(\R^3)}(T) \le 0,
$$
then $\|  \nabla \vec{v} \|^2_{L^2(\R^3)} (t)$
can only decrease for $t \ge T$.
This means that if the initial conditions verify the above inequality,
then $\|  \nabla \vec{v} \|^2_{L^2(\R^3)} (t)$ will decrease for all times and the existence of a global
a (unique) smooth solution is ensured.
Stated otherwise, for small enough initial conditions, such that
$$
\|   \vec{v} \|^{2}_{L^2(\R^3)}(t=0)\|  \nabla \vec{v} \|^{2}_{L^2(\R^3)}(t=0) \le \frac{\nu}{2 C_\nu}
$$
then a global solution  exists.
Note also that each time we have proved the existence of a smooth solution, proving the uniqueness
is easy. This follows from the application of the Gronwall's theorem to the difference of two possible solutions
if we assume that there exists two solutions. Then, we can show that this difference is zero, which proves the uniqueness.

Finally, the case of Euler's equations is quite similar. Indeed, the kinetic energy is conserved and thus independent of the time,
and, for what concerns $\|  \nabla \vec{v} \|^2_{L^2(\R^3)}$, we can derive an inequality which is close
to the one obtained in lemma \ref{lemma2}. We obtain
\begin{equation}
 \frac{1}{2} \frac{d}{dt}  \|  \nabla \vec{v} \|^2_{L^2(\R^3)} \le
 \left | \int_{\R^3} \partial_i v_j \partial_i v_k \partial_k v_j dV  \right |.
\label{dv22euler}
\end{equation}
From equation (\ref{dv22euler}), similarly to the Navier-Stokes case, it is possible to prove the local existence of a smooth  solution of Euler's equations, 
given smooth initial conditions.

After having derived general properties, the question
\ref{nonrel2} refers to the physical meaning of these mathematical relations.
This is why it is an essential question, without which the mathematical view is useless.
How can we build explicitly blowing up or not solutions of Navier-Stokes equations?
A key point here is to realize that Navier-Stokes equations are invariant under the
transformation
\begin{eqnarray}
\vec{v}_\lambda(\vec{x},t)&=&\lambda\vec{v}(\lambda\vec{x},\lambda^2t) \label{lambda1} \\
P_\lambda(\vec{x},t)&=&\lambda^2 P(\lambda\vec{x},\lambda^2t) .
\end{eqnarray}

This means that if the velocity field $\vec{v}(x,y,z,t)$ 
is a solution of the equations, then the velocity field $\vec{v}_\lambda(x,y,z,t)$ 
is another
acceptable solution (by construction).
Then, the following equality holds trivially
$$
 \| \vec{v}_{1/\lambda} \|^2_{L^2(\R^3)}= \lambda \| \vec{v} \|^2_{L^2(\R^3)}.
$$
We can think of this transformation, with $\lambda \gg 1$, as taking the fine scale behavior of the velocity field 
$\vec{v}(x,y,z,t)$and matching it with an identical (but rescaled and slowed down) coarse scale
behavior $\vec{v}_{1/\lambda}(x,y,z,t)$.
Previously, we were searching for bounds on $\| \vec{v} \|^2_{L^2(\R^3)}(t)$ and $\| \nabla \vec{v} \|^2_{L^2(\R^3)}(t)$.
Let us assume that such bounds
exist. We label them respectively as M and C. Obviously, for the transformed velocity field $\vec{v}_{1/\lambda}(x,y,z,t)$,
these bounds are multiplied by $\lambda$ (here assumed to be larger than $1$). 
This means that the bounds are worsened.
This also means that each time we have a solution of Navier-Stokes equations with
bounds M and C, then another solution
is possible, with worsened bounds scaled by $\lambda$. Blow-up can occur when
a solution of Navier-Stokes equations shifts its energy into increasingly finer and finer scales, thus evolving
more and more rapidly and eventually reaching a singularity in which the scales (in both space and time) tend
towards zero. In such a configuration, we lose obviously the effectiveness of the bounds (and consequently the
control) on the maximum energy and cumulative energy dissipation. For example, this is possible that at some
time, a solution of the equations shifts its energy from a spatial scale $1/\lambda$ to $1/(2\lambda)$ in a time of order $1/\lambda^2$,
providing that the concept of fluid element still holds.
Then,
if this behavior repeats over and over again, this is clear that the solution is divergent. This simple argument
gives an intuitive view of the problem.
It is quite easy to build a solution of Navier-Stokes equations that is smooth over a finite period of time
but that will blow-up at some point. In order to prevent this, we need to find a natural way to stop the above
cascade phenomenon.

\begin{conjecture}
From the discussion of questions \ref{nonrel1} and \ref{nonrel2}, we can think that a solution to 
the Millennium problem would  be to show that it is always possible to construct mathematically
a solution that diverges in a finite time, following the cascade behavior (above), apart from some specific initial conditions like small initial data where the global regularity is ensured.
Then, global regularity would not exist in general and the answer to the question of the Millennium problem would be negative.
\end{conjecture}

We can now move to question \ref{nonrel3} which is correlated to the above discussion in the sense as it reverses
the charge of the proof. 
The first step is to find a new way of bounding the
term  $-\int_{\R^3} \partial_i v_j \partial_i v_k \partial_k v_j dV$ in equation (\ref{dv22}), in order to
derive consequently
another bound for $\|  \nabla \vec{v} \|^2_{L^2(\R^3)}$.
After some algebra, we can show  that there exists a positive constant $A$ such that
$$
-\int_{\R^3} \partial_i v_j \partial_i v_k \partial_k v_j dV
\le
A \|  \Delta \vec{v} \|^{15/8}_{L^2(\R^3)}
 \|         \vec{v} \|^{1/8}_{L^2(\R^3)}
 \|         \vec{v} \|_{L^4(\R^3)},
$$
where the $L^4(\R^3)$ norm of the velocity field is apparent. This term reads
$$
\|         \vec{v} \|_{L^4(\R^3)}(t) = \left [ \int_{\R^3} \| \vec{v}(x,y,z,t) \|^4 dV \right]^{1/4}.
$$
We already know that the term in $\|         \vec{v} \|^{1/8}_{L^2(\R^3)}$ is bounded at all times by its initial value.
In addition, if we assume that the $L^4(\R^3)$ norm of the velocity field is bounded at all times, then there exists a 
constant $K > 0$ such that 
$$
-\int_{\R^3} \partial_i v_j \partial_i v_k \partial_k v_j dV
\le
K \|  \Delta \vec{v} \|^{15/8}_{L^2(\R^3)}
$$
From a scaled Young's inequality, we can transform this relation into
$$
-\int_{\R^3} \partial_i v_j \partial_i v_k \partial_k v_j dV
\le
\frac{\nu}{2} \|  \Delta \vec{v} \|^{2}_{L^2(\R^3)} + F_\nu,
$$
where $F_\nu$ is a constant that depends on the $K$ and $\nu$.
Then, using equation (\ref{dv22}) we obtain
$$
 \frac{1}{2} \frac{d}{dt}  \|  \nabla \vec{v} \|^2_{L^2(\R^3)} 
+
\frac{\nu}{2} \|  \Delta \vec{v} \|^{2}_{L^2(\R^3)}
\le
F_\nu.
$$
The last step is to use the Poincar\'e's inequality
$$
\|  \Delta \vec{v} \|^{2}_{L^2(\R^3)} \ge C \|  \nabla \vec{v} \|^2_{L^2(\R^3)},
$$
where $C$ is a positive constant.
Finally
$$
 \frac{d}{dt}  \|  \nabla \vec{v} \|^2_{L^2(\R^3)} 
+
\frac{\nu}{C}  \|  \nabla \vec{v} \|^{2}_{L^2(\R^3)}
\le
2 F_\nu.
$$
This  implies immediately that $\|  \nabla \vec{v} \|^2_{L^2(\R^3)}(t)$ is also bounded at all times and thus
the existence of smooth a (unique) global solution is ensured, under the assumption
that the $L^4(\R^3)$ norm of the velocity field is bounded at all times.
Intuitively, we can understand that a universal $L^4(\R^3)$ norm of the velocity would help for 
building a smooth (not blowing up) solution. Indeed, after the transformation (\ref{lambda1}), we have
$$
 \| \vec{v}_{1/\lambda} \|^4_{L^4(\R^3)}= \frac{1}{\lambda} \| \vec{v} \|^4_{L^4(\R^3)}.
$$
Hence, a bound on the $L^4(\R^3)$ norm of the velocity field becomes increasingly better for
$\vec{v}_{1/\lambda}$ when the parameter $\lambda$ is increased.
Under such conditions, it is possible to understand intuitively from the discussion
of question \ref{nonrel2} why blow-up do not occur.
Unfortunately, there is no way to prove that the  $L^4(\R^3)$ norm of the velocity field is bounded at all times
based solely on Navier-Stokes equations. Let us note that it is possible to define slightly modified equations that 
meet this requirement \cite{grafke}.

For completeness, let us conclude this part by another way of treating the global regularity of the velocity field, 
through a physical approach.
Indeed, the physical process responsible for the mathematical difficulties of the 
incompressible Navier-Stokes equations (in 3 dimensions
of space) is vortex stretching, related to the vorticity $\vec{\omega} = \vec{\nabla} \times \vec{v}$.
Vorticity is twice the local angular velocity of a fluid element. It follows the property
$$
\|  \vec{\omega}(.,t) \|^{2}_{L^2(\R^3)}=\|  \nabla \vec{v} (.,t) \|^{2}_{L^2(\R^3)}.
$$
This identity implies that we can not obtain any new results, not derived by using the velocity field or its first gradient, when using the vorticity.
What we can do is to get a renewed physical view of the open questions. 
Using  vorticity,
the incompressible Navier-Stokes equations can be written  as
\begin{equation}
	\frac{\partial}{\partial t} \vec{\omega}+ \left(\vec{v}\cdot\nabla\right)\vec{\omega}  = \nu \Delta \vec{\omega}
+\left(\vec{\omega}\cdot\nabla\right)\vec{v},
\label{vorticityeq}
\end{equation}
where the gradient of the pressure has disappeared when taking the rotational of the Navier-Stokes equations (\ref{ns}).
The above dynamical (vector) equation for the vorticity (\ref{vorticityeq}) expresses that the angular acceleration of a
fluid element (left hand side of the equation) results from the diffusive exchange of angular momentum with neighboring 
elements (term in $\nu \Delta \vec{\omega}$) and vortex stretching (term in $\left(\vec{\omega}\cdot\nabla\right)\vec{v}$).
We speak of vortex stretching because, at locations where $\left(\vec{\omega}\cdot\nabla\right)\vec{v}\cdot \vec{\omega} > 0$,
the fluid element is stretched in the direction of $\vec{\omega}$ and compressed in at least one other orthogonal direction. Then, the moment
of inertia of the fluid element is decreasing. By conservation of angular momentum, this means that the speed of rotation 
is increasing. This phenomenon competes with the diffusion term and can contribute to the amplification of
$\|  \vec{\omega}(.,t) \|^{2}_{L^2(\R^3)}$. 
Also,
this process can work in reverse to suppress  vorticity in other
locations. In two dimensions of space, the vortex stretching term is absent from equation (\ref{vorticityeq}). This is related to
the vanishing of the integral $\int_{\R^2} \partial_i v_j \partial_i v_k \partial_k v_j dS$ defined on $\R^2$ (already mentioned above).
Then, there can be obviously no problem with the global existence of smooth solutions when smooth initial conditions are given.
This discussion suggests that the study of the physical vortex dynamics (for itself), specifically local vortex stretching mechanisms, could be an interesting 
 approach to the Millennium problem.

\section{On the role of relativistic hydrodynamics in heavy-ion collisions}
\label{hi}
In this section, we examine fluid mechanics from another point of view: the short time evolution of 
a relativistic system, that can be produced in 
high energy heavy-ion collisions
(section \ref{rel}).
The applicability of hydrodynamics requires a system in local equilibrium (section \ref{hyp}), that is, the system 
has to behave as a complete entity like a liquid or a gas rather than a collection of individual particles.
These particles must be interacting with each other to reach equilibrium. The question is then whether their interactions
are frequent enough for thermodynamic equilibrium to be established. 
At RHIC, the energy of a collision is about $100$ GeV per nucleon. This means that each incoming nucleus is contracted 
by a Lorentz factor $\gamma \simeq 100$. The collision creates thousands of particles in a small volume. Thus,
 this is not impossible that this system may reach a state of local thermodynamic equilibrium.
One way to quantify the frequency of collisions
is by comparing the mean free path $\lambda_{mfp}$, the average distance a particle travels between two collisions, and
the typical size $L$ of the medium.
The mean free path is defined as 
\begin{equation}
\lambda = \frac{1}{n \sigma},
\label{mfp}
\end{equation}
where $n$ is the particle density and $\sigma$ the interaction cross section. A local thermodynamic
equilibrium can be reached only in the regime 
$\lambda \ll L$ (section \ref{hyp}).
In this context, a high energy heavy-ion collision would lead to the formation of a very dense (and hot) medium, with a mean free path
much smaller than the nuclear radius. This system would reach local thermodynamic equilibrium quickly (in a bout $1$ fm/c) and then
starts to undergo an   hydrodynamic expansion  due in particular to huge pressure gradients in the medium. 
There are respectively the steps (i) and (ii) of section \ref{rel}.
This is during the
stage (ii), which would last for about $10$ fm/c that the dense nuclear medium is expected to exist is a specific 
state of matter named Quark Gluon Plasma (QGP) \cite{Huovinen:2013wma,Ollitrault:2008zz}.
Then, after the medium has expanded and cooled sufficiently, the lower densities force
matter to decay into the observed hadrons, which corresponds to the last step (iii).


The idea to to evaluate the ratio $K=\lambda_{mfp}/L$, called the Knudsen number, and to check
whether the condition $K \ll 1$ can be realized. 
Below, we discuss the following proposition.
\begin{proposition}
The condition $K \ll 1$ can be realized only under the condition that a medium is formed 
(quickly after the heavy-ion collision) by strongly coupled constituents.
\end{proposition}
\begin{proof}
We can not prove rigorously this statement. Our purpose is mainly to discuss the ideas.
The only way to compute $K$ is to assume that the thermodynamic equilibrium is
obtained and that subsequent thermodynamic relations can be written. For simplicity, we use the laws derived in the massless gas limit.
First, we consider a weakly coupled system. This means that the interaction cross section $\sigma$  appearing in equation (\ref{mfp}) can
be expressed as 
\begin{equation}
\sigma \sim \frac{g^4}{<p>^2}, 
\label{xsp}
\end{equation}
where $g$ is the coupling constant ($g \ll 4 \pi$) and $<p>$ the average momentum of particles.
In the massless gas limit, $<p>=3 T$, where $T$ is the temperature of the medium. Also the particle density $n$ is related to
the entropy density by $s=4n$, easily computable in the massless gas limit. 
Altogether, we obtain 
\begin{equation}
K=\frac{1}{n \sigma L} \simeq \frac{4}{ g^4 L T}
\end{equation}
With the typical values for the system considered here, $L \simeq 10$ fm, $T \simeq 200$ MeV and $g \ll 2 \pi$, the Knudsen 
number $K$ is then necessarily greater than $0.4$, which does not justify the hypothesis of thermodynamic local equilibrium.
As a second step, we relax the weak coupling approximation. Then, the cross section $\sigma$ does not follow 
the expression (\ref{xsp}). Indeed, we can not give any interesting formula for $\sigma$. 
However, for a strongly coupled system, we can expect that there are a lots of collisions and, on average, each collision brings the
system closer to equilibrium. Hence, intuitively, it is more probable to find the system in thermodynamic equilibrium than
in the weakly coupled case. We make this statement more quantitative below.
We
start from the basic definition of the Knudsen number $K=\lambda_{mfp}/L$ and find a lower bound for $\lambda_{mfp}$.
Quite generally in the physics case of a strongly coupled system, the mean free path needs to be larger that the De Broglie 
wave length of the system.
Then
\begin{equation}
K \le \frac{1}{<p> L} \simeq \frac{1}{ 3 L T} \sim 0.01,
\end{equation}
provided that the relation $<p> = 3T$ holds approximately. Of course, this can only be a very crude approximation.
Even the definition of the mean free path is not clear for a strongly coupled system as the number of particles
is  not well defined. Another critical point is that
it is by no means obvious that the equilibration 
can be achieved, since it is a process which may take longer
time than the life time of the system under consideration. In general we assume that this can be reached quickly after
the collision, but we need to keep in mind that this is only an assumption.
This discussion, even with a lot of unknowns, gives some elements of knowledge concerning question \ref{rel1}.
\end{proof}

We now assume that the state of local thermodynamic equilibrium is reached (quickly after the collision) at some time $t_0$. Hence, the subsequent evolution 
follows the laws of relativistic hydrodynamics, provided that initial conditions are specified.
The idea is to solve the set of coupled equations (\ref{enrel}) and (\ref{eulerrel}) or (\ref{enrel2}) and (\ref{eulerrel2}),
complemented with an equation of state. For simplicity, we first discuss the case where viscosity effects can be neglected.
Of course, there is a first problem linked to the equation of state, that we do not know. We do not discuss this point
in this document. We always assume a linear equation of the form $P =\gamma \epsilon$.
The problem linked to initial conditions is that this is very unclear how to specify these  conditions in
the experimental context described here. This is question \ref{rel2}.
A complete set of initial conditions involves the three component of the velocity field, the energy and entropy density, at each point
in space and at time $t_0$. We know that
if $t_0$ is short enough, the transverse
components $v_x$ and $v_y$ of the fluid velocity are zero, 
due to the isotropy required by the existence of the local thermodynamic equilibrium.
Indeed, 
isotropy implies that there is no preferred direction, and that
the transverse momentum averaged over a fluid element vanishes. 

Beyond that point, we need effective models in order to define initial conditions.
We discuss the main ideas below.
A complete review can be found in \cite{review}. Also,  it is important to separate precisely the consequences of the hydrodynamic flow from those of the initial and boundary conditions. From that point of view, the effective models need to provide  simplified
pictures which can be qualitatively understood in physical terms. 

One model poses a complete stopping of the initial ions, that release their initial energy
in the volume corresponding to their longitudinal Lorentz-contracted size
in the direction of motion ($z$-axis) of the initial ions times the interaction area. Thus the velocity field
at initial time is zero for the whole volume.
Under the influence of the longitudinal gradient, the system starts expanding and the early expansion 
can be regarded as one-dimensional.
This is known as the Landau's initial conditions.  
In this situation, the one-dimensional expressions of equations (\ref{enrel}) and (\ref{eulerrel}) with zero baryon chemical potential
can be solved. It can be shown that this leads to a Gaussian distribution of the particle rapidity distribution.
As the system expands, the mean free path  increases until $K \sim 1$. 
At this stage, the
particles depart from each other so fast that the collision processes become ineffective. 
The system breaks up and the particles stream away from each other (freeze-out).
We note that this model presents the
simplification  that
the evolution of the system before freeze-out is dominated by the longitudinal motion and thus, 
 the hydrodynamic transverse motion can
be neglected or at least factorized out. 
A refined model for initial conditions has been proposed by Bjorken, based on the experimental fact that
the rapidity distribution of particles is constant in the mid-rapidity region.
This means that the central region is invariant under Lorentz transformations along the direction of motion.
This implies that all quantities of interest characterizing the central region depend only on the
longitudinal proper time $\tau=\sqrt{t^2-z^2}$ and transverse coordinates $x$ and $y$.
Then, following equations (\ref{enrel2}) and (\ref{eulerrel2}) the evolution equation of the the energy density becomes
\begin{equation}
\frac{d\epsilon}{d\tau}=-\frac{\epsilon+P}{\tau}.
\label{eps2}
\end{equation}
Using a simple form $P=\gamma \epsilon$ for the equation of state, it can be shown that the 
energy density is decreasing with the proper time as
\begin{equation}
\epsilon(\tau,x,y)=\epsilon(\tau_0,x,y) \left ( \frac{\tau_0}{\tau}  \right )^{1+\gamma}
\label{enbj}
\end{equation}
Consequently, the temperature is also decreasing with the proper time as 
\begin{equation}
T(\tau,x,y)=T(\tau_0,x,y) \left ( \frac{\tau_0}{\tau}  \right )^{\gamma}.
\end{equation}
Equation (\ref{enbj}) is useful in practice once we define $\epsilon(\tau_0,x,y)$.
This can be done again with an effective model. The most common one is called Glauber model,
which is essentially a way to encode geometrical constraints on the expression of the energy density and similarly on the entropy density \cite{review}.
The general idea of this model is that Woods-Saxon distributions of nuclear matter in colliding heavy ions are
projected on the transverse plane ($x,y$), and the resulting densities projected on this plane together with the
nucleon-nucleon cross section (at the collision energy) are used to calculate the number density
of binary collisions and participants, that have interacted at least once, on the transverse plane.
Then, the energy density $\epsilon(\tau_0,x,y)$ is taken to be proportional to the profile of collisions or participants or to
a linear combination of them. The proportionality constant is a free parameter chosen to reproduce the
observed final particle multiplicity \cite{review}.
Similarly for the entropy density $s(\tau_0,x,y)$.
What is really striking is that the Glauber model, which relies only on the nuclear geometry and straight light propagation
of nucleons at high energy, provides a perfect description of a wide range of spectra of experimental variables.
In experiments, the Glauber model is used to make the link between the experimental variables like the transverse energy)deposited in zero degree calorimeters after the heavy-ion collision and the physical conditions that prevail for this collision,
like its centrality, the initial energy or entropy densities.
This is the answer to  question \ref{rel3}.
The geometry of a given heavy-ion collision is one essential information, linked to the knowledge of initial conditions
\cite{bialas}.
However, for the full interpretation of experimental measurements, geometry alone is obviously not sufficient and needs  
to be completed the hydrodynamical equations, as discussed above. Indeed,
non-central collision will lead to spatial asymmetries of the initial state produced. This is geometry. Then, these asymmetries will be mapped
directly into the final state momentum distribution thanks to the hydrodynamical evolution equations.
This is the dynamics.

Finally, let us discuss briefly the viscosity corrections in the relativistic context
and why they do not change qualitatively the previous discussion.
First, the viscosity of a fluid is related to its ability to return to local thermodynamic equilibrium.
For a medium with strongly interacting constituents at the beginning of the 
hydrodynamical evolution, we can expect small relaxation times, thus
small viscosity effects. However, during the expansion, as the mean free path increases, the coupling
strength is decreased between constituents. Therefore, we can not claim that viscosity effects should be small in average
during the expansion.
Quite generally,  viscosity  can be encoded following a standard strategy by modifying
the inviscid energy-momentum tensor
$T^{\mu\nu}$. Then, the inviscid  expression of the tensor is complemented by
terms depending on the first derivatives of the velocity field. 
These terms correspond to $1/L$ corrections to the inviscid expression, where $L$ is the
typical size of the system. This gives the equivalent to the Navier-Stokes equations in the
non-relativistic case. 
Other terms with second derivatives of the velocity field would correspond to $1/L^2$ corrections,
and thus more complex equations.
However, the global behavior of the energy density as a function of $\tau$ described by the equation (\ref{enbj}) remains 
\cite{review,reviewviscousrel}.

\section{Conclusions and outlook}

In this paper, we have discussed some questions that we think to be important concerning fluid mechanics
in its relativistic limit or not.
We have ranged them in two categories: the long term behavior of the velocity field for non-relativistic fluids, where initial
and boundary conditions are known, and 
the short time evolution of 
relativistic systems, that can be produced in 
high energy heavy-ion collisions at RHIC or LHC experiments.
In this last case, initial and boundary conditions represent the critical points.

For non-relativistic fluids, when we assume smooth initial conditions
of the velocity field
(divergence free, infinitely differentiable, square integrable and with strong decay properties at infinity),
the main question is whether the solution of the three-dimensional incompressible
Navier-Stokes equations can develop a finite time singularity or if a global smooth 
solution exist, which means smooth for all times. This question is also known as the Millennium problem related to Navier-Stokes
equations. Following the arguments that we have developed in this paper, we can think that a way to approach  
the problem would  be to show that it is always possible to construct mathematically
a solution that diverges in a finite time, apart from some specific initial conditions like small initial data where the global regularity is ensured.
Then, global regularity would not exist in general.
We have also shown why the study of the physical vortex dynamics, specifically local vortex stretching mechanisms, could be an interesting approach to the Millennium problem.

Concerning relativistic systems, that can be produced in 
high energy heavy-ion collisions, we have shown that the realization of the local thermodynamic equilibrium, necessary
for any subsequent application of hydrodynamical equations, is already a problematic issue. We have proved
qualitatively why it can be obtained only for
strongly coupled systems. This is not something intuitive. Indeed, we could have thought that the state of 
matter created during the heavy-ion collision  would be  a weakly coupled gas of (almost free) quarks and gluons in thermal equilibrium, due to the property of asymptotic freedom at high energy. However, we have shown that this last statement is 
not correct. Then, comes the next difficult item concerning the initial (and boundary) conditions. Indeed, we there is not direct way 
to determine the initial conditions that prevail at the formation of the state of matter produced quickly after the collision,
as we can only observe hadrons formed in the late stage evolution of the system.
Here, we have shown that the geometry of the collision that can be inferred using for example the Glauber model is an essential input to
the initial conditions. Moreover, it happens that the pure geometrical constraints allow a very good understanding of a lot of
experimental observations. Then, using the equations of fluid mechanics that transformed the initial spatial asymmetries (geometry)
into the final state momentum asymmetries is a way to a global understanding of experimental measurements possible for heavy-ion collisions.


\end{document}